\documentclass{llncs}
\usepackage{makeidx}  % allows for indexgeneration
\usepackage{tikz}
\usetikzlibrary{arrows, positioning, chains}
\usepackage{subfig}
\usepackage{pgfplots}
\pgfplotsset{compat=1.10}
\usepackage{amssymb, amsmath, amsfonts}
\usepackage{algorithm}
\usepackage{algorithmic}
\usepackage{environ}

\usepackage[square,comma,numbers,sort&compress]{natbib}
\usepackage{natbibspacing}
\usepackage[]{tocbibind} 

\begin{document}
%
%\frontmatter          % for the preliminaries
%
\pagestyle{headings}  % switches on printing of running heads
%\addtocmark{Hamiltonian Mechanics} % additional mark in the TOC
%

\mainmatter              % start of the contributions
\title{Author Obfuscation using Generalised Differential Privacy}
\titlerunning{Author Obfuscation}  % abbreviated title (for running head)
\author{Natasha Fernandes, Mark Dras, Annabelle McIver}
\authorrunning{Natasha Fernandes et al.} % abbreviated author list (for running head)
\institute{Macquarie University, Sydney, Australia} 

\maketitle              % typeset the title of the contribution

\begin{abstract}
The problem of obfuscating the authorship of a text document has received little attention in the literature to date. Current approaches are ad-hoc and rely on assumptions about an adversary's auxiliary knowledge which makes it difficult to reason about the privacy properties of these methods. Differential privacy is a well-known and robust privacy approach, but its reliance on the notion of adjacency between datasets has prevented its application to text document privacy. However, generalised differential privacy permits the application of differential privacy to arbitrary datasets endowed with a metric and has been demonstrated on problems involving the release of individual data points. In this paper we show how to apply generalised differential privacy to author obfuscation by utilising existing tools and methods from the stylometry and natural language processing literature.
\keywords{generalised differential privacy, author obfuscation, word mover's distance}
\end{abstract}
%
%-------------------------------- INTRODUCTION -------------------------------------

\section{Introduction}
The proliferation of machine learning techniques and publicly available datasets has resulted in some exciting advances in data analytics. At the same time, some well-publicised privacy breaches, notably the AOL and Netflix examples, have created concerns about data privacy and the ability for privacy methods to protect against machine learning attacks. Differential privacy is a key privacy definition which has rapidly gained popularity due to its mathematical foundations and, importantly, its independence from assumptions about the external data sources available to adversaries. This makes it an important tool for protection of personal data in the face of machine learning adversaries trained on large-scale datasets.

Differential privacy relies on a notion of an `individual' in a dataset, under the assumption that datasets are structured into rows of individuals, and adversaries are agents which query the data for statistical information. These notions are not naturally transferable to unstructured datasets such as text documents.

Generalised differential privacy is an extension of differential privacy which can be applied to arbitrary datasets endowed with a metric. This permits its application to more general datasets, and has found most application in geo-location privacy, involving the differentially private release of users'

We recall the usual definition of differential privacy which says that, for any pair of adjacent datasets $x, x'$ and some output $z$ from a query, a (probabilistic) mechanism $K$ satisfies $\epsilon$-differential privacy if
\[
             K(x)(z) \leq e^\epsilon K(x')(z)
\]
for some non-negative $\epsilon$.~\footnote{We use only the strict $\{\epsilon, 0\}$ version of differential privacy in this paper, and not its relaxation ($\epsilon$, $\delta$)-differential privacy.}

Generalised differential privacy extends this notion to domains endowed with a metric as follows: for any domain of secrets $\mathcal{X}$ endowed with a metric $d_{\mathcal{X}}$ and any elements $x, x' \in \mathcal{X}$, a (probabilistic) mechanism $K$ satisfies $\epsilon d_{\mathcal{X}}$-privacy if
\[
             K(x)(z) \leq e^{\epsilon d_{\mathcal{X}}(x, x')} K(x')(z)
\]

In this paper we show how generalised differential privacy can be applied to author obfuscation. This task requires the private release of documents so as to protect the identity of the author whilst maintaining some semantic properties of the document. We draw on existing notions of authorship from the stylometry and natural language processing literature which incorporate the use of distance measures between authors. 

%------------------------------------------- BACKGROUND --------------------------------------------

\section{Problem Description}

Author obfuscation is the task of obscuring a piece of text in order to hide its authorship whilst preserving its semantic content. Stylometric techniques have identified three types of features used to identify authorship of a document: stylistic, word-based and character-based features. Stylistic features are typically lexical, syntactic or document-level characteristics. For example, average word length, average sentence length and frequency of use of particular words are all features which can be unique for authors.  Word-based methods treat each word in the document as a feature, and represent a document as a \emph{bag of words}, which ignores word ordering but preserves frequency counts of individual words. Finally, character-based features treat individual sequences of characters as features for document representation. These sequences are referred to as \emph{character n-grams}. For example, the character 3-gram representation of the phrase "There it is" would be `The', `her', `ere', `re\_', `e\_it', `it\_', `t\_i', `\_is'~\footnote{Note that we use `\_' to represent spaces.}.

\subsection{Model}

We envisage an author who wishes to release a document which preserves the topicality of the original document whilst masking stylometric features which may reveal their authorship. By 'topicality' we refer to the ability of a document consumer to identify the original topic of the document. In order to provide a privacy guarantee for any adversary (ie over any prior), we will use generalised differential privacy. Our privacy promise is that the output document is almost the same, regardless of whether the input document was $x$ or some 'close' document $x'$, where closeness is defined using an appropriate metric distance. In order to achieve privacy over authors, we need to modify the features in the document which reveal authorship. This can be done firstly by considering documents formatted as bags of words (BOW). Such document are useful for machine learning applications, which typically use BOW formats and ignore word ordering. This formatting also removes word ordering information, which reveals important stylistic information. Secondly, we can remove `stopwords', which are words such as pronouns and prepositions, which do not contribute to the meaning of the document. These have been shown to be highly effective in author attribution, but because they contribute nothing to our utility requirement they can be safely discarded. Finally, we can consider the removal of words which do not significantly contribute to the topicality of the document, so as to reduce the document down to the smallest set of words which guarantee some usefulness. This can be done using a feature classifier to identifier the most significant features for topic classification.

\subsection{Application of Differential Privacy}

We can represent words as real-valued vectors using a word embedding representation such as Word2Vec or GloVe. These representations preserve semantic relationships between words, where the semantic distance is typically measured using either cosine similarity or Euclidean distance. This also allows the entire word embedding vocabulary to also be treated as a synonym set for any word in the vocabulary. A natural metric to then consider for measuring the semantic distance between documents is the Word Mover's Distance. This metric has been designed specifically for use with word embedding vectors, and is based on the well-known Earth Mover's Distance. Briefly, the Word Mover's Distance is the cost of moving all the words from one document to the words in another document. We will formalise this in the next section.

%---------------------------------- PRELIMINARIES --------------------------------
\section{Preliminaries}

In this section we lay out some definitions for use throughout the rest of this paper. 
We will only be interested in discrete sets so we present a simplified formalisation, noting that the definition also applied to continuous sets and distributions.

Let $\mathcal{X}$ and $\mathcal{Z}$ be finite sets and let $\mathbb{P}(\mathcal{Z})$ be the set of probability measures over $\mathcal{Z}$. We define a \emph{mechanism} as a probabilistic function $K: \mathcal{X} \rightarrow \mathbb{P}(\mathcal{Z})$ 

Recall that a metric $d: \mathcal{X} \times \mathcal{X} \rightarrow [0, \infty)$ satisfies (i) $d(x,y) = 0$ iff $x=y$, (ii) $d(x, y) = d(y, x)$ and (iii) $d(x, y) + d(y, z) \ge d(x, z)$ for all $x, y, z \in \mathcal{X}$. 
We denote by $d_2$ the Euclidean metric on $\mathbb{R}^n$. 

We recall the definition of generalised differential privacy:

\begin{definition}{(Generalised Differential Privacy)}
\label{gdp}
Let $\epsilon > 0$. A mechanism $K : \mathcal{X} \rightarrow \mathbb{P}(\mathcal{Z})$ satisfies $d_\mathcal{X}$-privacy, iff $\forall x, x' \in \mathcal{X}$: 
\[
         K(x)(Z) \leq e^{\epsilon d_\mathcal{X}(x, x')} K(x')(Z) \quad \forall Z \subseteq \mathcal{Z}~\footnote{We make the simplifying assumption that all sets in $\mathcal{Z}$ are measurable.}
\]
\end{definition}

We now formalise some notions from the natural language processing literature. 

Let $\mathcal{V}$ be a fixed finite vocabulary of words from all possible documents.  A \emph{bag of words} (BOW) is an unordered, finite-length lists of words from $\mathcal{V}$ with duplicates permitted. A \emph{document vector} is an ordered bag of words.  A \emph{word embedding vector} is a $k$-dimensional real-valued vector representing a word in $\mathcal{V}$, for some fixed positive integer $k$. We assume the existence of a word embedding vector lookup table, denoted $W$, such that $W(w)$ returns the word embedding vector for the word $w \in \mathcal{V}$. 

We denote by $\mathcal{W}$ the universe of word embedding vectors (for all words in $\mathcal{V}$). We denote by $\mathcal{U}(\mathcal{W})$ the set of unordered lists of word vectors, also known as BOWs, and by $\mathcal{O}(\mathcal{W})$ the set of ordered lists of word vectors, also known as document vectors.

We note that we can transform a BOW into a document vector by fixing an (arbitrary) ordering of words.

The Word Mover's Distance can be formally defined as follows:

\begin{definition}{(Word Mover's Distance)}
Let $x, y \in \mathcal{O}(\mathcal{W})$ be document vectors of lengths $a$ and $b$ respectively. We assume the existence of a non-negative, real-valued cost function over $\mathcal{V}$. Let $C \in \mathbb{R}^{a \times b}$ be a cost matrix, where $C_{ij}$ represents the cost of moving word $i$ in $x$ to word $j$ in $y$. Define $T \in \mathbb{R}^{a \times b}$ to be a flow matrix where the entry $T_{ij}$ denotes how much of word $i$ in $x$ moves to word $j$ in $y$. Then the Word Mover's Distance $d_W(x, y)$ is defined as the solution to the linear optimisation problem:
\[
     d_W(x, y) = \min\limits_{T \ge 0} \sum_{i, j} T_{ij} C_{ij} 
\]
    subject to:
\[
   \sum_j T_{ij} = \frac{1}{a} \qquad \forall i \in [1 \ldots a]
\]
and
\[
    \sum_i T_{ij} = \frac{1}{b}  \qquad \forall j \in [1 \ldots b]
\]
\end{definition}

We also define our notion of document privacy under the term \emph{document-indistinguishability}.

\begin{definition}{(Document-Indistinguishability)}\label{DI}
Let $d_W:  \mathcal{U}(\mathcal{W}) \times \mathcal{U}(\mathcal{W}) \rightarrow [0, \infty)$ be the Word Mover's Distance defined on BOW documents. A mechanism $K: \mathcal{U}(\mathcal{W}) \rightarrow \mathbb{P}(\mathcal{U}(\mathcal{W}))$ satisfies $\epsilon$-document-indistinguishability iff for all $x, x' \in \mathcal{U}(\mathcal{W})$ :
\[
          K(x)(Z) \leq e^{\epsilon d_W(x, x')} K(x')(Z)\quad \forall Z \subseteq \mathcal{U}(\mathcal{W})
\]
\end{definition}

\section{Privacy Using the Word Mover's Distance}

We now present some results on the Word Mover's Distance which will allow us to produce a differentially private mechanism for documents.

\subsection{Optimal Solution to the Word Mover's Distance}

Our first result shows that an optimal solution to the Word Mover's Distance problem involves the movement of whole words only, precisely when the source and destination documents are the same length. In order to prove this result, we first introduce some results on doubly stochastic matrices.

\begin{definition}{}
An $n \times n$ matrix whose elements are non-negative and has all rows and columns summing to 1 is called \emph{doubly stochastic}. A doubly stochastic matrix which contains only 1's and 0's is called a \emph{permutation matrix}. 
\end{definition}

\begin{theorem}{(Birkhoff-von Neumann)}
The set of $n \times n$ doubly stochastic matrices forms a convex polytope whose vertices are the $n \times n$ permutation matrices.
\end{theorem}

The Birkhoff-von Neumann theorem says that the set of doubly stochastic matrices is a closed, bounded convex set, and every doubly stochastic matrix can be written as a convex combination of the permutation matrices. We can use this theorem to prove the following result.

\begin{theorem}{}
Let $C$ be an $n \times n$ cost matrix. Then the optimisation problem $\text{minimise} \sum\limits_{i, j} T_{ij} C_{ij}$ where $T$ is an $n \times n$ doubly stochastic matrix always has an $n \times n$ permutation matrix as an optimal solution.
\label{ds_opt}
\end{theorem}

We are now ready to present a result on the Word Mover's Distance, namely that documents of equal length always have an optimal solution which does not involve partial movements of words. We note that the order of words in a document does not affect the calculation of the Word Mover's Distance, hence the following result on document \emph{vectors} also holds for BOW documents.

%-------------------------------- MAIN THEOREM -------------------------------------

\begin{theorem}{}
Let $d_1, d_2 \in \mathcal{O}(\mathcal{W})$ be $n$-dimensional document vectors. Then the Word Mover's Distance $d_W(d_1, d_2)$ has an optimal solution involving only the movement of whole words. 
\label{wmd-thm}
\end{theorem}

\begin{proof}
Since $d_1$ and $d_2$ are both of length $n$, the flow matrix $T$ must be $n \times n$ so we can rewrite the constraints as
\[
             \sum\limits_i T_{ij} = \frac{1}{n} \quad \forall j \in [1 \ldots n] \quad \text{and} \quad \sum\limits_j T_{ij} = \frac{1}{n} \quad \forall i \in [1 \ldots n]
\] 
Notice that the optimisation problem is in essence unchanged if we multiply $T$ by a constant, so we could also write
\[
             \sum\limits_i T_{ij} = 1 \quad \forall j \in [1 \ldots n] \quad \text{and} \quad \sum\limits_j T_{ij} = 1 \quad \forall i \in [1 \ldots n]
\] 

Thus we have a $T$ that is doubly stochastic. From Theorem~\ref{ds_opt} the optimal solution includes a permutation matrix, that is, a matrix in which there is exactly one 1 in each row and column and the remaining elements are 0. But this corresponds to a flow where each word in $d_1$ moves entirely to a whole word in $d_2$. This completes the proof.
\qed
\end{proof}

%-------------------------- EXAMPLE 2 -------------------------------

\subsubsection{Example} 
We will use a simple example to demonstrate how the flow matrix can be simplified when both documents are of the same length. Consider the documents $\vec{d} = $ `Obama speaks Illinois' and $\vec{d'} = $ `President greets press'. The relative mass of each word in each document is $\frac{1}{3}$ since both documents are the same length; this is depicted by the vectors $\vec{d_p}$ and $\vec{d'_p}$ in Figure~\ref{wmd2}. We imagine a cost matrix such that the flow matrix $T$ given in Figure~\ref{wmd-exa} is optimal. Now, consider the scenario in Figure~\ref{wmd-exb} where the relative mass of each word is $1$. Since we have simply multiplied the relative weights of all words by 3 without changing the cost matrix, the optimal solution matrix $T'$ will correspond exactly to $3T$, and becomes doubly stochastic. Although the computed distance will also be 3 times the original distance, the \emph{relative} flow between words in the 2 documents is preserved. In particular, whole word flows in Figure~\ref{wmd-exa} correspond with whole word flows in Figure~\ref{wmd-exb}. Thus the result of Theorem~\ref{wmd-thm} stands in both cases.

%----------------------------- FIGURE 2 -------------------------------

\begin{figure}%
\small
\centering
\begin{minipage}{0.4\linewidth}%
  \begin{tikzpicture}[>=stealth,
        every node/.style={block},
        block/.style={minimum height=2em,minimum width=3em,outer sep=0pt,draw,rectangle,node distance=0pt}]
        \node (A) [] {$\frac{1}{3}$};
        \node (A1) [right=of A] {$\frac{1}{3}$};
        \node (A2) [right=of A1] {$\frac{1}{3}$};    
        
        \node [draw=none, left=0.3cm of A] {$\vec{d_p} =$};    
        
        \node (C) [below=0.6cm of A]  {$\frac{1}{3}$};
        \node (C1) [right=of C] {$\frac{1}{3}$};
        \node (C2) [right=of C1] {$\frac{1}{3}$};  
        
        \node [draw=none, left=0.3cm of C] {$\vec{d'_p} =$}; 
        
        \draw [->, dashed] (A) edge node [draw=none] {} (C1);
        \draw [->, dashed] (A1) edge node [draw=none] {} (C);
        \draw [->, dashed] (A2) edge node [draw=none] {} (C2);  
\end{tikzpicture}%
\vskip 0.5cm
\begin{tikzpicture}[>=stealth,
        every node/.style={block},
        block/.style={minimum height=2em,minimum width=2em,outer sep=0pt,rectangle,node distance=0pt}]
        \node (T11) [] {$0$};
        \node (T12) [right=of T11] {$\frac{1}{3}$};
        \node (T13) [right=of T12] {$0$};  
        
        \node (T21) [below=of T11]  {$\frac{1}{3}$};
        \node (T22) [right=of T21] {$0$};
        \node (T23) [right=of T22] {$0$};  
        
        \node (T31) [below=of T21]  {$0$};
        \node (T32) [right=of T31] {$0$};
        \node (T33) [right=of T32] {$\frac{1}{3}$};  
                
        \node [draw=none, left=0.6cm of T21] {\textbf{T} =}; 
        
        \draw [decorate, decoration={brace,raise=3pt}] (T13.north east) -- (T33.south east) node (a) [midway,draw=none] {};
        \draw [decorate, decoration={brace,mirror,raise=3pt}] (T11.north west) -- (T31.south west) node (b) [midway,draw=none] {};                      
\end{tikzpicture}%
\caption{\small Sample vectors and flow matrix $T$ for standard Word Mover's Distance formulation.}%
\label{wmd-exa}%
\end{minipage}%
\qquad%
\begin{minipage}{0.4\linewidth}%
\begin{tikzpicture}[>=stealth,
        every node/.style={block},
        block/.style={minimum height=2em,minimum width=3em,outer sep=0pt,draw,rectangle,node distance=0pt}]
        \node (A) [] {$1$};
        \node (A1) [right=of A] {$1$};
        \node (A2) [right=of A1] {$1$};    
        
        \node [draw=none, left=0.3cm of A] {$\vec{d_p} =$};    
        
        \node (C) [below=0.6cm of A]  {$1$};
        \node (C1) [right=of C] {$1$};
        \node (C2) [right=of C1] {$1$};  
        
        \node [draw=none, left=0.3cm of C] {$\vec{d'_p} =$}; 
        
        \draw [->, dashed] (A) edge node [draw=none] {} (C1);
        \draw [->, dashed] (A1) edge node [draw=none] {} (C);
        \draw [->, dashed] (A2) edge node [draw=none] {} (C2);  
\end{tikzpicture}%
\vskip 0.5cm
\begin{tikzpicture}[>=stealth,
        every node/.style={block},
        block/.style={minimum height=2em,minimum width=2em,outer sep=0pt,rectangle,node distance=0pt}]
        \node (T11) [] {$0$};
        \node (T12) [right=of T11] {$1$};
        \node (T13) [right=of T12] {$0$};  
        
        \node (T21) [below=of T11]  {$1$};
        \node (T22) [right=of T21] {$0$};
        \node (T23) [right=of T22] {$0$};  
        
        \node (T31) [below=of T21]  {$0$};
        \node (T32) [right=of T31] {$0$};
        \node (T33) [right=of T32] {$1$};  
                
        \node [draw=none, left=0.6cm of T21] {\textbf{T'} =}; 
        
        \draw [decorate, decoration={brace,raise=3pt}] (T13.north east) -- (T33.south east) node (a) [midway,draw=none] {};
        \draw [decorate, decoration={brace,mirror,raise=3pt}] (T11.north west) -- (T31.south west) node (b) [midway,draw=none] {};                      
\end{tikzpicture}%
\caption{\small Modified vectors using multiplicative factor of 3 to produce doubly stochastic flow matrix $T'$.}%
\label{wmd-exb}%
\end{minipage}%
\caption{\small Example of modifying the Word Mover's Distance problem to generate a doubly stochastic flow matrix. Multiplying the relative word weights by a constant corresponds to multiplication of the flow matrix by the same constant, however the relative flows between words remains unchanged.}
\label{wmd2}
\end{figure}

%---------------------------- INTRO TO MAIN THEOREM 2 -----------------------------

\subsection{Extension to Document-Indistinguishability}

We are now ready to present our main theorem, which provides a connection between generalised differential privacy for vectors and document indistinguishability.

\begin{theorem}{}
Let $K: \mathbb{R}^k \rightarrow \mathbb{P}(\mathbb{R}^k)$ be a mechanism operating on real-valued $k$-dimensional vectors which satisfies $\epsilon d_2$-privacy. Then $K$ can be extended to a mechanism $K^*: \mathcal{U}(\mathcal{W}) \rightarrow \mathbb{P}(\mathcal{U}(\mathcal{W}))$ operating on documents which satisfies $\epsilon$-document-indistinguishability.
\label{d2-thm}
\end{theorem}

\subsubsection{Proof Sketch}

Firstly, we can think of $K$ as a mechanism operating on document (ordered) vectors, and show that it can be extended to a mechanism operating on BOW documents (unordered vectors) by considering the appropriate permutations. It turns out that we can fix any ordering of words in the source document and just consider permutations of the output document. This means the privacy guarantee for input documents $x, x'$ is determined by any ordering we choose. We can therefore choose an ordering which minimises the transportation cost between corresponding words in the documents, which, from Theorem~\ref{wmd-thm}, corresponds to the WMD. 

\section{Privacy Mechanism}

We are now ready to present a privacy mechanism for document-indistinguishability. We have seen that to find a mechanism satisfying $\epsilon$-document-indistinguishability, it suffices to find a mechanism operating on vectors which satisfies $\epsilon d_2$-privacy. Previous work~\citep{andres2013geo} has shown how this can be done in 2-dimensions via the planar Laplacian. We now present an extension of this result to n-dimensions. 

\subsection{n-Dimensional Laplace Mechanism}

Given sets $\mathcal{X}, \mathcal{Z}$ of n-dimensional vectors, we would like a mechanism $K$ with pdf $D$ satisfying
\[
   D(x)(z) \propto e^{- \epsilon d_2(x, z)} \quad \text{for } x \in \mathcal{X}, z \in \mathcal{Z}
\] 

Such a mechanism is called a Laplace mechanism and satisfies $\epsilon d_2$-privacy~\citep{chatzikokolakis2013broadening}. We require a method of selecting a vector according to this distribution.

Noting that $D$ is spherically symmetric, and using translation invariance, we can consider the distribution $D(0)(z)$ and translate this by $x$ to get the distribution $D(x)(z)$. For notational convenience we write $D(0)(z)$ as $D_0(z)$.

Using $z = (z_1, z_2, \ldots, z_n)$ then $D_0(z)$ can be rewritten as 
\[
    D_0(z) = ce^{- \epsilon \sqrt{z^2_1 + z^2_2 + \ldots + z^2_n}}
\]

Calculating the constant $c$ can be done by a mapping $(z_1, z_2, \dots, z_n) \mapsto (r, \theta_1, \theta_2, \dots, \theta_{n-1})$ to spherical co-ordinates. This yields the following integral
\begin{align*}
  \int\limits_{0}^{\infty} c_1r^{n-1}e^{-\epsilon r} dr \int\limits_0^\pi c_2\sin^{n-2}\theta_1 d\theta_1 \int\limits_0^\pi c_3\sin^{n-3}\theta_2 d\theta_2 \dots \int\limits_0^{\pi} c_{n-1}\sin\theta_{n-2} d\theta_{n-2} \int\limits_0^{2\pi} c_n d\theta_{n-1}
\end{align*}
where $\prod\limits_{k=1}^n c_k = c$.

This is a product of independent distributions, and in particular, the constant $c_1$ in the first integral can be shown to evaluate to $\frac{\epsilon ^n}{(n-1)!}$ yielding
\[
 \int\limits_{0}^{\infty} \frac{\epsilon ^n}{(n-1)!} r^{n-1} e^{-\epsilon r} dr = \int\limits_{0}^{\infty} \frac{\epsilon ^n}{\Gamma(n)} r^{n-1} e^{-\epsilon r} dr
\]
which we recognise as the PDF of the Gamma distribution
\[
     f(x, k; \theta) = \frac{1}{\Gamma(k) \theta^k} x^{k-1} e^{-\frac{x}{\theta}}
\]
where $x \mapsto r$, $k \mapsto n$ and $\theta \mapsto \frac{1}{\epsilon}$.

The remaining product of integrals is equivalent to selecting an n-dimensional vector uniformly over the unit n-sphere. Therefore, the selection of a random n-dimensional Laplace vector can be achieved by selecting a random vector uniformly over the surface of an n-sphere and applying a scaling factor drawn from the gamma distribution. We formalise this as follows
\begin{theorem}{}
Let $U_{S_1}$ be the uniform distribution on the $n$-dimensional sphere of radius 1, and denote by $\text{Gamma}(k, \theta)$ the Gamma distribution with shape $k$ and scale $\theta$. Let $K: \mathbb{R}^n \rightarrow \mathbb{P}(\mathbb{R}^n)$ be a mechanism operating on real-valued $n$-dimensional vectors which outputs $z$ with distribution $x + UR$, where $U \sim U_{S_1}$ and $R \sim \text{Gamma}(n, \frac{1}{\epsilon})$. Then the mechanism $K$ is a Laplace mechanism satisfying $\epsilon d_2$-privacy.
\label{e2-thm}
\end{theorem}

Note that choosing n=2 results in the planar Laplacian described in \cite{andres2013geo}.

Several methods have been proposed for the generation of random variables from the Gamma distribution~\citep{kroese2013handbook} as well as the uniform selection of points on the unit n-sphere~\citep{marsaglia1972choosing}. We will not present methods for the former, as there are already implementations in libraries such as Scipy (for Python) for selecting from the Gamma distribution. However, there is a nice method for selecting a random vector from the surface of the unit n-sphere which has been described previously in the literature~\citep{marsaglia1972choosing}. The method is to select $n$ random variables from the Gaussian distribution over $[0, 1]$ into an $n$-dimensional vector $v$ and output the normalised vector $\frac{v}{|v|}$. This method allows a random unit $n$-vector to be drawn without requiring a transformation from polar co-ordinates.

\subsection{Mechanism for Document Indistinguishability}

We now present a mechanism satisfying $\epsilon$-indistinguishability. We assume that input documents have first been converted into a bag of words with stopwords discarded. We then use the method for generating fixed-length documents described in \citep{weggenmann2018syntf}. That is, the bag of words document can be used to generate a distribution over words using the frequency count of each word in the document. A fixed-length document can be generating by selecting $n$ words from the document according to the distribution. These preprocessing steps d

This is shown in Algorithm~\ref{obf2}.

\begin{algorithm}
\caption{\small Obfuscation Mechanism}\label{obf2}
\small
\begin{algorithmic}
\REQUIRE{ epsilon $\epsilon$, word embeddings $W$, documents $d$}
\FOR{doc in d} \STATE{
    words = list words in doc\;
    \FOR{w in words} \STATE{
          x = lookup vector for w in $W$\; \\
          r = select scale according to Gamma(dim(x), $\frac{1}{\epsilon}$)\; \\
          u = select unit vector uniformly on dim(x)-sphere\; \\
          z = x + ru\; \\
          z' = lookup closest word to z in $W$\; \\
          add z' to noisy\_doc }
    \ENDFOR \\
    add noisy\_doc to obfuscated dataset
}

\ENDFOR \\
\RETURN {obfuscated dataset}
\end{algorithmic}
\end{algorithm}

\begin{theorem}{}
The mechanism presented in Algorithm~\ref{obf2} satisfies $\epsilon$-document-indistinguishability.
\end{theorem}

\begin{proof}
The inner loop contains the Laplace mechanism as described in Theorem~\ref{e2-thm}, which satisfies $\epsilon d_2$-privacy. The 'closest word' step represents a post-processing truncation of the vector $z$, which does not change the $\epsilon d_2$-privacy guarantee of the inner loop. The outer loop applies mechanism in the inner loop to every word in the document. By Theorem~\ref{d2-thm} this outer loop satisfies $\epsilon$-document-indistinguishability.
\qed
\end{proof}

\appendix

\section{Proofs Omitted from Section 4}

\begin{theorem}{}
Let $C$ be an $n \times n$ cost matrix. Then the optimisation problem $\text{minimise} \sum\limits_{i, j} T_{ij} C_{ij}$ where $T$ is an $n \times n$ doubly stochastic matrix always has an $n \times n$ permutation matrix as an optimal solution.
\end{theorem}

\begin{proof}
We prove this by contradiction. Let $T^*$ be an optimal $n \times n$ solution matrix. We know that such a solution exists by the Birkhoff-von Neumann Theorem (since the set of solutions is closed and bounded). We assume firstly that $T^*$ is not a permutation matrix, and secondly that no permutation matrix is optimal. Let $\{P^1, P^2, \ldots, P^k\}$ be the set of $n \times n$ permutation matrices. Then, by the Birkhoff-von Neumann theorem, we can write 
\begin{align}
    T^* = \lambda_1 P^1 + \lambda_2 P^2 + \ldots + \lambda_k P^k
\label{eqn-bvn}
\end{align}
where $\lambda_i \ge 0$ and $\sum\limits_{i = 1}^k \lambda_i = 1$. Since $T^*$ is optimal and none of the $P^i$ are optimal, we can also write
\begin{align*}
              \sum\limits_{i,j} P^m_{ij}C_{ij} &> \sum\limits_{i,j} T^*_{ij}C_{ij} \quad \text{(by assumption)}                                                    
\end{align*}
for $0 < m \le k$. And thus we have
\begin{align*}
           \sum\limits_{i,j} T^*_{ij} C_{ij} &= \sum\limits_{i,j} (\lambda_1 P^1_{ij} + \ldots + \lambda_k P^k_{ij}) C_{ij} \quad \text{(from~\ref{eqn-bvn})} \\
                                       &= \sum\limits_{i,j} \lambda_1 P^1_{ij} C_{ij} +  \ldots + \sum\limits_{i,j} \lambda_k P^k_{ij} C_{ij} \quad \text{(factorising)} \\
                                       &>  \sum\limits_{i,j} \lambda_1 T^*_{ij} C_{ij} + \ldots + \sum\limits_{i,j} \lambda_k T^*_{ij} C_{ij} \quad \text{(by assumption)} \\
                                       &= \lambda_1  \sum\limits_{i,j} T^*_{ij} C_{ij} + \ldots + \lambda_k \sum\limits_{i,j} T^*_{ij} C_{ij} \quad \text{(arithmetic)} \\
                                       &= (\lambda_1 + \ldots + \lambda_k) \sum\limits_{i,j} T^*_{ij} C_{ij} \quad \text{(factorising)} \\
                                       &= \sum\limits_{i,j} T^*_{ij} C_{ij} \quad \text{(since $\sum\limits_{i = 1}^k \lambda_i = 1$)} %\quad \text{\Lightning} 
\end{align*}
which is a contradiction. Thus, either $T^*$ is a permutation matrix, or there must be a permutation matrix which is also optimal.
\end{proof}

The following lemma is useful in proving the next main result on document indistinguishability.

%-------------------------------------- LEMMA ---------------------------------------

\begin{lemma}{}
Let $K: \mathcal{W} \rightarrow \mathbb{P}(\mathcal{W})$ be a mechanism operating on word vectors and let $d, z \in \mathcal{U}(\mathcal{W})$ be documents of length $n$. Then $K$ can be extended to a mechanism $K^*: \mathcal{U}(\mathcal{W}) \rightarrow \mathbb{P}(\mathcal{U}(\mathcal{W}))$ operating on documents such that
\[
    K^*(d)(z) = \sum_i K(w_1)(v_{\phi_i(1)}) \times K(w_2)(v_{\phi_i(2)}) \times \ldots \times K(w_n)(v_{\phi_i(n)})
\]
where the $w_i, v_i$ are words (arbitrarily labelled) in $d, z$ respectively, the $\phi_i$ are permutation functions, and the sum is over unique permutations of words in $z$. % TODO: Need to define phi properly
\label{lemma-extend}
\end{lemma}

\begin{proof}
Choose an arbitrary ordering of words in $d$ and $z$ and let $\vec{d} = <w_1, w_2, \ldots, w_n>$, $\vec{z} = <v_1, v_2, \ldots, v_n>$ be the corresponding document vectors. Let $\phi_i: S \rightarrow S$ be a sequence of permutation functions over $S = \{1, \ldots, n\}$ for $i = \{1, \ldots, m\}$ such that $<v_{\phi_i(1)}, v_{\phi_i(2)}, \ldots, v_{\phi_i(n)}>$ is a unique permutation of words in $z$ for each $i$.

Now, it is straightforward to extend $K$ to a mechanism $K': \mathcal{O}(\mathcal{W}) \rightarrow \mathbb{P}(\mathcal{O}(\mathcal{W}))$ operating on document \emph{vectors}, since we can simply apply $K$ to each word in order. That is,
\[
     K'(\vec{d})(\vec{z}) = K(w_1)(v_1) \times K(w_2)(v_2) \times \ldots \times K(w_n)(v_n)
\]

Clearly, $K'(\vec{d})$ defines a valid probability distribution for any $\vec{d}$ since we sum over all possible output vectors $\vec{z}$.

In order to extend this to a mechanism over documents, observe that the mechanism $K'$ produces the same output distribution regardless of the ordering of words in the document vector $\vec{d}$ (since the mechanism $K$ operates on each word independently). Therefore we only need to consider permutations of words in the output document vector $\vec{z}$. The distribution over documents is then given by the sum of distributions over each permutation of words in the output vector $\vec{z}$, that is,
\[
    K^*(d)(z) = \sum_i K(w_1)(v_{\phi_i(1)}) \times K(w_2)(v_{\phi_i(2)}) \times \ldots \times K(w_n)(v_{\phi_i(n)})
\]

Clearly $K^*(d)$ also defines a valid probability distribution, since it produces the same distribution as $K'(\vec{d})$ except that the output probabilities are `collected' for all permutations of the output vector.
\qed
\end{proof}

\begin{theorem}{}
Let $K: \mathbb{R}^k \rightarrow \mathbb{P}(\mathbb{R}^k)$ be a mechanism operating on real-valued $k$-dimensional vectors which satisfies $\epsilon d_2$-privacy. Then $K$ can be extended to a mechanism $K^*: \mathcal{U}(\mathcal{W}) \rightarrow \mathbb{P}(\mathcal{U}(\mathcal{W}))$ operating on documents which satisfies $\epsilon$-document-indistinguishability.
\end{theorem}

\begin{proof}
Since we have a real-valued vector representation for words, we can treat $K$ as a mechanism operating on word vectors. Let $\vec{d} = <w_1, w_2, \ldots, w_n>$, $\vec{z} = <v_1, v_2, \ldots, v_n>$ be $n$-dimensional document vectors and let $\phi_i: S \rightarrow S$ be a sequence of permutation functions over $S = \{1, \ldots, n\}$ for $i = \{1, \ldots, m\}$. From Lemma~\ref{lemma-extend}, we can extend $K$ to a mechanism $K^*: \mathcal{U}(\mathcal{W}) \rightarrow \mathbb{P}(\mathcal{U}(\mathcal{W}))$ satisfying
\[
    K^*(\vec{d})(\vec{z}) = \sum_i K(w_1)(v_{\phi_i(1)}) \times K(w_2)(v_{\phi_i(2)}) \times \ldots \times K(w_n)(v_{\phi_i(n)})
\]

We can choose any particular ordering of words in $\vec{d}$ since the ordering of words is arbitrary. Fix any ordering of words in $\vec{d}$ and let $\vec{d'} = <w'_1, w'_2, \ldots, w'_n>$ be an $n$-dimensional document vector, where the word order in $\vec{d'}$ is chosen to minimise the sum of the Euclidean distances between corresponding words in $\vec{d}$ and $\vec{d'}$. That is, we choose an ordering which minimises $\sum\limits_{i=1}^n d_2(w_i, w'_i)$.

Now, for the document vectors $\vec{d}$, $\vec{d'}$, we have that
\begin{align}
    \frac{K^*(\vec{d})(\vec{z})}{K^*(\vec{d'})(\vec{z})} &= \frac{ \sum_i K(w_1)(v_{\phi_i(1)}) \times K(w_2)(v_{\phi_i(2)}) \times \ldots \times K(w_n)(v_{\phi_i(n)}) }{ \sum_i K(w'_1)(v_{\phi_i(1)}) \times K(w'_2)(v_{\phi_i(2)}) \times \ldots \times K(w'_n)(v_{\phi_i(n)}) }  \nonumber \\
              &= \frac{ \sum_i \prod_j K(w_j)(v_{\phi_i(j)})}{ \sum_i \prod_j K(w'_j)(v_{\phi_i(j)})} 
    \label{vecs}
\end{align}
where the number of terms in the numerator and denominator is the same (since this depends only on $\vec{z}$).
But we also have that
\begin{align}
    K(w_k)(v_{\phi_i(k)}) \le K(w'_k)(v_{\phi_i(k)}) e^{\epsilon d_2(w_k, w'_k)}  \quad \text{($\epsilon d_2$-privacy)} \label{d2}
\end{align}
for all words in $\vec{d}$, $\vec{d'}$ and all permutations $\phi_i$.
Therefore,
\begin{align}
    \frac{K^*(\vec{d})(\vec{z})}{K^*(\vec{d'})(\vec{z})} &= \frac{ \sum_i \prod_j K(w_j)(v_{\phi_i(j)}) }{ \sum_i \prod_j K(w'_j)(v_{\phi_i(j)}) } \quad \text{(from ~\ref{vecs})} \nonumber \\
                        &\le \frac{ \sum_i \prod_j K(w'_j)(v_{\phi_i(j)}) e^{\epsilon d_2(w_j, w'_j)} }{\sum_i \prod_j K(w'_j)(v_{\phi_i(j)}) } \quad \text{(from ~\ref{d2})}  \nonumber \\
                        &= \frac{ \sum_i e^{\epsilon (d_2(w_1, w'_1) + \ldots + d_2(w_n, w'_n))} \prod_j K(w'_j)(v_{\phi_i(j)}) }{ \sum_i \prod_j K(w'_j)(v_{\phi_i(j)})} \quad \text{(arithmetic)}  \nonumber \\
                        &= \frac{ e^{\epsilon (d_2(w_1, w'_1) + \ldots + d_2(w_n, w'_n))} \sum_i \prod_j K(w'_j)(v_{\phi_i(j)})}{ \sum_i \prod_j K(w'_j)(v_{\phi_i(j)})} \quad \text{(arithmetic)}  \nonumber \\
                        &= e^{\epsilon (d_2(w_1, w'_1) + \ldots + d_2(w_n, w'_n))} \quad \text{(cancelling like terms)}
\label{final}
\end{align} 

Now, notice that the documents $\vec{d}$ and $\vec{d'}$ have the same dimension (necessarily, due to the operation of the mechanism $K^*$). Therefore we know from Theorem~\ref{wmd-thm} that the Word Mover's Distance $d_\mathcal{W}(\vec{d}, \vec{d'})$ has an optimal solution involving the movement of whole words. That is, there exists a permutation of words $<w'_{\phi_i(1)}, w'_{\phi_i(2)}, \ldots, w'_{\phi_i(n)}>$ in $\vec{d'}$ such that $d_\mathcal{W}(\vec{d}, \vec{d'}) = \sum\limits_k d_2(w_k, w'_{\phi_i(k)})$. But we chose an ordering of words in $\vec{d'}$ that minimises $\sum\limits_k d_2(w_k, w'_k)$. Recalling that the Word Mover's Distance is minimal, we therefore must have that 

\[
    d_\mathcal{W}(\vec{d}, \vec{d'}) = \sum\limits_k d_2(w_k, w'_k)
\]

And so,
\begin{align*}
   \frac{K^*(\vec{d})(\vec{z})}{K^*(\vec{d'})(\vec{z})} &\le e^{\epsilon (d_2(w_1, w'_1) + \ldots + d_2(w_n, w'_n))} \quad \text{(from ~\ref{final})} \\
             &= e^{\epsilon d_W(\vec{d}, \vec{d'})}
\end{align*}
Thus the mechanism $K^*$ satisfies $\epsilon$-document-indistinguishability.
\qed
\end{proof}

\section{Proofs Omitted from Section 5}

We present here a more complete proof of the derivation of the $n$-dimensional Laplace mechanism. 

We consider the distribution $D_0(z) = ce^{- \epsilon \sqrt{z^2_1 + z^2_2 + \ldots + z^2_n}}$ for some constant $c$.

In order to select a point from this distribution, we consider the CDF
\begin{align}
\label{cdf}
    F(z) = \idotsint\limits_{Z_A} ce^{- \epsilon \sqrt{z^2_1 + z^2_2 + \ldots + z^2_n}} \,dz_1 \dots dz_n
\end{align}
for some region of interest $Z_A$.

To compute this we require a change of co-ordinates. We can convert from Cartesian co-ordinates to spherical co-ordinates 
\[
    (z_1, z_2, z_3, \ldots, z_n) \mapsto (r, \theta_1, \theta_2, \ldots, \theta_{n-1})
\]

as stated in \citep{mustard1964} using a transformation $r = \sqrt{z^2_1 + z^2_2 + \ldots + z^2_n}$ with inverse
\begin{align*}
    z_1 &= r \cos\theta_1 \\
    z_2 &= r \sin\theta_1 \cos\theta_2 \\
    z_3 &= r \sin\theta_1 \sin\theta_2 \cos\theta_3 \\
    &\ldots \\
    z_{n-1} &= r \sin\theta_1 \sin\theta_2 \ldots \sin\theta_{n-2} \cos\theta_{n-1} \\
    z_n &= r \sin\theta_1 \sin\theta_2 \ldots \sin\theta_{n-2} \sin\theta_{n-1}     
\end{align*}

We also need to calculate the matrix of partial derivatives to get the Jacobian, which is well known to be
\[
    \frac{\partial(z_1, z_2, \ldots, z_n)}{\partial(r, \theta_1, \ldots, \theta_{n-1})} = r^{n-1} \sin^{n-2}\theta_1 \sin^{n-3}\theta_2 \ldots \sin^2\theta_{n-3} \sin\theta_{n-2}
\]

And therefore the integral in (\ref{cdf}) becomes

\begin{align*}
  &\idotsint\limits_{Z_A} ce^{- \epsilon \sqrt{z^2_1 + z^2_2 + \ldots + z^2_n}} \,dz_1 \dots dz_n \\
  &  \quad = \idotsint\limits_{Z_A} ce^{-\epsilon r} r^{n-1} \sin^{n-2}\theta_1 \sin^{n-3}\theta_2 \ldots \sin^2\theta_{n-3} \sin\theta_{n-2} dr d\theta_1 \dots d\theta_{n-1} \\
  &  \quad = \int\limits_{0}^{R} c_1r^{n-1}e^{-\epsilon r} dr \int\limits_0^\pi c_2\sin^{n-2}\theta_1 d\theta_1 \int\limits_0^\pi c_3\sin^{n-3}\theta_2 d\theta_2 \dots \int\limits_0^{\pi} c_{n-1}\sin\theta_{n-2} d\theta_{n-2} \int\limits_0^{2\pi} c_n d\theta_{n-1} \\
\end{align*}
where $\prod\limits_{k=1}^n c_k = c$.

We note that this is a product of independent distributions, and thus the co-ordinates (radius and angles) can be selected independently. We also require that each integral sums to 1 to get valid probability distributions, so firstly we can calculate the constant $c_1$ using

\begin{align}
\label{pdf}
   \int\limits_{0}^{\infty} c_1r^{n-1}e^{-\epsilon r} dr = 1
\end{align}

Using integration by parts we find
\[
     \int\limits_{0}^{\infty} c_1r^{n-1}e^{-\epsilon r} dr = \frac{c_1}{\epsilon^n} (n-1)!
\]

And thus the integral in (\ref{pdf}) becomes
\[
     \int\limits_{0}^{\infty} \frac{\epsilon ^n}{(n-1)!} r^{n-1} e^{-\epsilon r} dr = \int\limits_{0}^{\infty} \frac{\epsilon ^n}{\Gamma(n)} r^{n-1} e^{-\epsilon r} dr
\]
which we recognise as the PDF of the Gamma distribution
\[
     f(x, k; \theta) = \frac{1}{\Gamma(k) \theta^k} x^{k-1} e^{-\frac{x}{\theta}}
\]
where $x \mapsto r$, $k \mapsto n$ and $\theta \mapsto \frac{1}{\epsilon}$.

Now, we also note that the remaining integrals (over the angles $\theta_1, \dots \theta_{n-1}$) correspond to the Jacobian for the unit n-sphere. In other words, this is equivalent to the problem of selecting a point uniformly over the surface of the $n$-sphere. 

\backmatter
\bibliography{MathBib}
\end{document}